\documentclass{article}

\usepackage[utf8]{inputenc}
\RequirePackage{amsthm,amsfonts,amssymb,centernot,float,import,makeidx,subfiles,authblk}
\RequirePackage{graphicx}
\usepackage{array}
\usepackage{amsmath}
\usepackage{booktabs}
\usepackage{breqn}
\usepackage{cancel}
\usepackage{caption}
\usepackage{dsfont}
\RequirePackage{fix-cm}
\usepackage{float}
\usepackage[hidelinks]{hyperref}
\usepackage{lmodern}
\usepackage{longtable}
\usepackage{lscape}
\usepackage{multirow}
\usepackage{threeparttable}
\usepackage{tikz}
\usepackage{varioref}
\usepackage{xcolor}
\usepackage{ulem,lipsum}
\newcommand\xoutpars[1]{\let\helpcmd\xout\parhelp#1\par\relax\relax}
\newcommand\soutpars[1]{\let\helpcmd\sout\parhelp#1\par\relax\relax}
\long\def\parhelp#1\par#2\relax{%
  \helpcmd{#1}\ifx\relax#2\else\par\parhelp#2\relax\fi%
}

\usetikzlibrary{positioning, calc, shapes.geometric, shapes, shapes.multipart, arrows.meta, arrows, decorations.markings, external, trees}

\usepackage[]{ulem}

\usepackage{bookmark}
\makeatletter
\providecommand*{\toclevel@titlech}{0} 
\edef\toclevel@authorch{\the\numexpr\toclevel@titlech+1} 

\makeatother

\tikzstyle{Arrow} = [
	thick, 
	decoration={
		markings,
		mark=at position 1 with {
			\arrow[thick]{latex}
			}
		}, 
	shorten >= 3pt, preaction = {decorate}
	]

\theoremstyle{plain}

\newtheorem{theorem}{Theorem}

\newtheorem{proposition}{Proposition}

\usepackage{enumitem}
\theoremstyle{remark}
\newtheorem{remark}{Remark}

\newtheorem{assumption}{Assumption}
\newtheorem{assumptionalt}{Assumption}[assumption]

\newenvironment{assumptionp}[1]{
  \renewcommand\theassumptionalt{#1}
  \assumptionalt
}{\endassumptionalt}

\usepackage{caption}
\begin{document}

\author[1]{Max Rubinstein}
\author[2]{Maria Cuellar}
\author[3]{Daniel Malinsky}

\affil[1]{RAND Corporation; E-mail: mrubinstein@rand.org}
\affil[2]{Department of Criminology and Department of Statistics and Data Science, University of Pennsylvania; E-mail: mcuellar@sas.upenn.edu}
\affil[3]{Department of Biostatistics, Columbia University; E-mail: d.malinsky@columbia.edu}

\title{Mediated probabilities of causation}


\maketitle

\begin{abstract}
We propose a set of causal estimands that we call ``the mediated probabilities of causation.'' These estimands quantify the probabilities that an observed negative outcome was induced via a mediating pathway versus a direct pathway in a stylized setting involving a binary exposure or intervention, a single binary mediator, and a binary outcome. We outline a set of conditions sufficient to identify these effects given observed data, and propose a doubly-robust projection based estimation strategy that allows for the use of flexible non-parametric and machine learning methods for estimation. We argue that these effects may be more relevant than the probability of causation, particularly in settings where we observe both some negative outcome and negative mediating event, and we wish to distinguish between settings where the outcome was induced via the exposure inducing the mediator versus the exposure inducing the outcome directly. We motivate these estimands by discussing applications to legal and medical questions of causal attribution.
\end{abstract}

\textbf{Keywords:} Mediation analysis, probability of causation, machine learning, non-parametrics, causal inference

\section{Introduction}

The probability of causation is a quantity that has been proposed to answer the question: what is the probability that an observed outcome was caused by a specific exposure and not by something else? In other words, can an outcome be attributed to a specific exposure? This question is often of interest in the law. Suppose that a judge or jury must determine whether an individual claimant's cancer (e.g., non-Hodgkin's lymphoma) was caused by their exposure to a specific chemical (e.g., an herbicide). In other words, can the individual's cancer be attributed to this chemical exposure?

The probability of causation is therefore a parameter that conditions on the observed outcome and the exposure received. This contrasts to other commonly studied parameters, such as the average treatment effect, which conditions on nothing; the conditional average treatment effect, which conditions on pre-exposure covariates; or, the average effect of treatment on the treated, which conditions only on treatment. In the cancer attribution case, one may evaluate whether herbicide exposure causes this specific type of cancer in the population (e.g., estimating the ATE). However, the ATE only informs whether the chemical is harmful in general, and it does not answer how likely it is that \textit{a specific individual's} cancer was caused by the chemical exposure and not by something else (e.g. smoking or other chemicals). A line of research has focused on defining and estimating probabilities of causation \cite{lagakos1986assigned, tian2000probabilities, dawid2017probability, cuellar2020non, dawid2022effects}.\footnote{One subtlety of these analyses is whether the question of interest is about a specific individual or about a group of individuals. Researchers have proposed various approaches \cite{cuellar2020non, dawid2022effects, dawid2014fitting} to address this issue. Clearly, in a legal setting similar to the example above, the question is about a single individual who has cancer. However, transferring conclusions about a group to an individual is challenging \cite{dawid2017individual}. For this article, we therefore take the stance that we seek to estimate parameters defined by a group of individuals that share characteristics $X=x$. Therefore, if one seeks to know the probability of causation for a specific individual, one can estimate it for a group who shares characteristics $x$ with the individual.}

However, the probability of causation may be too coarse of a measure for some settings. Specifically, it fails to pinpoint any specific mechanism or hypothesized causal pathway that induces the outcome through a mediator. To illustrate, consider the recent lawsuit against Harvard University alleging discrimination against Asian American applicants.\footnote{https://www.nytimes.com/2022/12/02/us/asian-american-college-applications.html} The probability of causation ostensibly helps formalize the plaintiffs' discrimination claim. Specifically, we can imagine a world where a rejected Asian American applicant were not (perceived by the admissions officer to be) Asian American, and ask whether, in this counterfactual world, their application would have been accepted. However, we might also wonder by what mechanism might the applicant's identity cause the admissions decision. One hypothesized pathway posited in this case is that (perceived) Asian American racial identity led admissions officers to give worse subjective personality assessments, and in turn deny applications. From the plaintiffs' perspectives, admissions rejections via this ``indirect'' effect would be impermissible. On the other hand, there may exist other relevant pathways that may be considered legally or morally permissible or not, depending on one's interpretation of the relevant legal statutes or background ethical commitments. 

Other legal and scientific questions also hinge on whether some intermediate factor induces an outcome. For example, in jury selection, attorneys sometime question whether a juror's exclusion might be due to factors causally related to -- or ``downstream from'' -- race. In evaluating the harmful effect of some chemical exposure on health outcomes, scientists or advocates may posit a specific mediating element such as a negative physiological response -- for example, a change in hormone levels or gene expression -- that ultimately leads to disease or death. In clinical contexts, a mediating negative event may also arise in the context of attributing ``cause of death'': for example, someone may hypothesize that a patient's death was caused by an adverse reaction to vaccination via a cardiovascular mechanism, having observed post-vaccination myocarditis. However, the probability of causation does not distinguish between possible pathways.

We therefore introduce and examine a new class of causal estimands that we call the ``mediated probabilities of causation.'' These estimands separately quantify the probability of causation via direct and indirect pathways, and together sum to a ``total'' mediated probability of causation. Through considering the simplified setting of a single binary exposure, mediator, and outcome, our primary contribution is to define these causal estimands and outline assumptions sufficient to identify these estimands in terms of observed data, building on the existing work on natural effects and path-specific effects \cite{imai2010general, steen2017flexible, nabi2018estimation, malinsky2019potential} (noting that these effects decompose average treatment effects, and do not condition on the observed outcome). As a second contribution, we propose a doubly-robust estimation strategy that targets projections of the mediated probabilities of causation, extending the estimation method outlined in \cite{cuellar2020non}, who introduced this approach to estimate the probability of causation. This strategy allows for the use of modern non-parametric and machine learning methods that require minimal modeling assumptions, while still yielding root-n consistent and asymptotically normal estimates under relatively mild conditions.

The closest related work to ours is \cite{dawid2024bounding}, wherein the authors derive bounds for the probability of causation under mediation only when there is no direct effect and no confounding. Our work is more general in that we allow for the existence of direct effects. Furthermore, we provide a more flexible estimating procedure. One of the strengths of our approach is that we make explicit the identification assumptions sufficient to arrive at observable quantities that can be estimated from data. We provide discussion and intuition about which assumptions may be strong for some settings; however, whether these assumptions hold should be evaluated for each application on a case-by-case basis. Thus, our contributions are primarily theoretical. Nevertheless, we believe this theory may be relevant for applications similar in structure to the aforementioned legal or medical attribution questions, when the assumptions are deemed plausible and the requisite data is available.

Our paper proceeds as follows: in Section \ref{sec:background} we briefly review the probability of causation, including the formal definition, the typical identifying assumptions, and an expression for the identified estimand. In Section \ref{sec:identification} we propose the three causal estimands that we call the mediated probabilities of causation: the probability of indirect causation, the probability of direct causation, and the total mediated probability of causation. We outline a set of assumptions sufficient to identify these estimands in observed data and provide expressions for these identified estimands. In Section \ref{sec:estimation} we propose a non-parametric projection-based approach to estimate these quantities, following \cite{cuellar2020non}. In Section \ref{sec:simulations} we illustrate our proposed methods using a simulation study motivated by the Harvard discrimination case. In Section \ref{sec:discussion} we conclude with general discussion of the problem, limitations of our study, and possible areas for future research.

\section{Background}\label{sec:background}
We begin by reviewing the probability of causation, the identifying assumptions sufficient to estimate this quantity in observed data, and the identified data functional \cite{cuellar2020non}. 

Let $Y$ denote a binary variable where $Y = 1$ reflects a negative outcome. Let $A$ denote a binary exposure, and $X$ denote a d-dimensional covariate vector, which may include continuous variables. Let $Y(a)$ denote the potential outcome for an individual under exposure assignment $A = a$. $Y(1)$ will denote the outcome under exposure. We can define the probability of causation as:

\begin{align}\label{eqn:pcause}
    \tau(x) &= P(Y(0) = 0 \mid Y(1) = 1, X = x)
\end{align}

\noindent This formalizes the probability that an individual with covariates $X=x$, who would have experienced the negative outcome under exposure, would \emph{not} have experienced the outcome under \emph{no} exposure. We do not observe both potential outcomes $Y(1)$ and $Y(0)$ in practice. We instead observe $n$ independent copies of the data $O = (Y, A, X)$. We can nevertheless identify $\tau(x)$ in observed data by making the following assumptions:

\begin{assumption}[Y consistency]\label{eqn:pc1}
    $Y = AY(1) + (1-A)Y(0)$
\end{assumption}

\begin{assumption}[A-Y ignorability]\label{eqn:pc2}
    $A \perp \{Y(0), Y(1)\} \mid X$
\end{assumption}

\begin{assumption}[A positivity]\label{eqn:pc3}
    $P\{\min_a P(A = a \mid X) \ge \epsilon \} = 1 \qquad \epsilon > 0$
\end{assumption}

\begin{assumption}[Y positivity]\label{eqn:pc4}
    $P\{P(Y = 1 \mid A = 1, X) \ge \epsilon \} = 1 \qquad \epsilon > 0$
\end{assumption}

\begin{assumption}[Y monotonicity]\label{eqn:pc5}
    $Y(1) \ge Y(0)$
\end{assumption}

Proposition \ref{prop:1} presents the identification result. For completeness, the proof is provided in Appendix \ref{app:proofs}.

\begin{proposition}\label{prop:1}
    Under Assumptions (\ref{eqn:pc1})-(\ref{eqn:pc5}), $\tau(x)$ is identified in the observed data distribution as,
    \begin{align*}
        \tau(x) &= 1 - \frac{P(Y = 1 \mid A = 0, X = x)}{P(Y = 1 \mid A = 1, X = x)}.
    \end{align*}
\end{proposition}

We may alternatively define a probability of causation that conditions on the observed outcomes and exposure assignment rather than the potential outcomes, since the counterfactual quantities $Y(1)$ are unobserved for observations where $A = 0$ and are therefore not of practical interest. Specifically, we may define the probability of causation as,

\begin{align}
\tilde{\tau}(x) &= P(Y(0) = 0 \mid Y = 1, A = 1, X = x).
\end{align}
    
\noindent It is easy to see that this expression is equivalent to $\tau(x)$ under assumptions (\ref{eqn:pc1})-(\ref{eqn:pc2}). However, when defining the probability of causation as $\tilde{\tau}(x)$, we can weaken the required identifying assumptions. For example, in place of A-Y ignorability it is sufficient to identify $\tilde{\tau}(x)$ that $A \perp Y(0) \mid X$. We nevertheless introduce the estimands conditional on the counterfactual quantities for conceptual clarity, noting throughout these alternative, and arguably more intuitive, formulations.

\section{Mediated probabilities of causation}\label{sec:identification}

To motivate the mediated probabilities of causation, consider the case where we observe a binary causal descendant of the exposure, denoted $M$, that mediates the effect of the exposure $A$ on the outcome $Y$, so that $O = (X, A, M, Y)$. Moreover, we let $M = 1$ denote some negative condition. Figure \ref{fig:dag} illustrates the assumed causal structure of the observed data distribution distribution.

\begin{figure}[H]
\caption{Assumed causal structure}
\begin{center}
\begin{tikzpicture}

[
array/.style={rectangle split, 
	rectangle split parts = 3, 
	rectangle split horizontal, 
    minimum height = 2em
    }
]
 \node (0) {X};
 \node [right =of 0] (1) {A};
 \node [right =of 1] (2) {M};
 \node [right =of 2] (3) {Y};

 \draw[Arrow] (0.east) -- (1.west);
 \draw[Arrow] (1.east) -- (2.west);
 \draw[Arrow] (2.east) -- (3.west);
 \draw[Arrow] (0) to [out = 25, in = 160] (2);
 \draw[Arrow] (0) to [out = 25, in = 160] (3);
 \draw[Arrow] (1) to [out = 25, in = 160] (3);
\end{tikzpicture}
\end{center}
\label{fig:dag}
\end{figure}
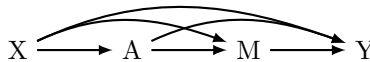

We next define potential outcomes to illustrate how information on the mediator adds complexity to this problem. We consider three specific estimands: the total mediated probability of causation, the probability of indirect causation, and the probability of direct causation. We motivate and describe each quantity in greater detail below.

\subsection{Total mediated probability of causation}

Viewing the mediator as an ``intermediate outcome'' possibly affected by exposure, we define a probability of causation on the stratum where $M(1) = 1$ and outcome $Y(1) = Y(1,M(1)) = 1$ (noting that $Y(a) = Y(a, M(a))$ here and throughout). In other words, we define a probability of causation where the potential mediator under exposure and potential outcome under exposure are both negative. We call this the ``total mediated probability of causation.'' Mathematically, we can express this as,

\begin{align}
\delta(x) = P(Y(0) = 0 \mid Y(1) = 1, M(1) = 1, X = x).
\end{align}

\noindent This estimand may be motivated by the question: given that an individual would experience a negative outcome and negative mediator under an exposure (i.e. $Y(1) = 1$ and $M(1) = 1$), what is the probability that the negative outcome would not have occurred absent the exposure? 

The quantity $\delta(x)$ as defined conditions on the potential rather than observed outcomes. Just as we could define the probability of causation conditional on only observed quantities, we can also define the total mediated probability of causation quantity conditional on the observed outcome, mediator, and exposure:
\begin{align}
\tilde{\delta}(x) = P(Y(0) = 0 \mid Y = 1, M = 1, A = 1, X = x)
\end{align}

\noindent This quantity expresses the probability that an individual who received the exposure and experienced a negative mediator and outcome under exposure would not have experienced the negative outcome had they not been exposed. These two quantities, $\delta(x)$ and $\tilde{\delta}(x)$, are equivalent under the identifying assumptions that we outline in Section \ref{sec:identification} below.

Importantly, the total mediated probability of causation is not generally equal to the probability of causation. Instead, the probability of causation can be expressed as a weighted average of the total mediated probability of causation and an analogous term defined on the stratum where $M(1) = 0$. We can see this by iterating expectations for $\tau(x)$ over the conditional distribution of $M(1)$:
\begin{align*}
    \tau(x) &= \sum_{m = 0, 1} P(Y(0) = 0 \mid Y(1) = 1, M(1) = m, X = x) \\
    \nonumber&\times P(M(1) = m \mid Y(1) = 1, X = x) \\
    &= \delta(x)\xi(x) + \delta'(x)(1 - \xi(x))
\end{align*}

\noindent where $\xi(x) = P(M(1) = 1 \mid Y(1) = 1, X = x)$. This expression highlights the relationship between these quantities, and shows that we should not in general expect the probability of causation to equal the total mediated probability of causation. We further discuss the term $\delta'(x)$ and other possible decompositions in Section \ref{ssec:estdisc} and in Appendix \ref{app:othres}.

\subsection{Probabilities of direct and indirect causation}

The total mediated probability of causation decomposes into the sum of what we call the probabilities of direct and indirect causation. Specifically, we iterate expectations of $\delta(x)$ over the conditional distribution of the cross-world counterfactual quantity $Y(1, M(0))$ to obtain these quantities.

\begin{align*}
\delta(x) &= \sum_{y = 0, 1} P(Y(1, M(0)) = y, Y(0, M(0)) = 0 \mid Y(1, M(1)) = 1, M(1) = 1, X = x) \\
&= \psi(x) + \zeta(x)
\end{align*}

\noindent where
\begin{align}\label{eqn:pnie}
\psi(x) = P(Y(1, M(0)) = 0, Y(0, M(0)) = 0 \mid Y(1, M(1)) = 1, M(1) = 1, X = x)
\end{align}

\noindent and
\begin{align}\label{eqn:pnde}
\zeta(x) &= P(Y(1, M(0)) = 1, Y(0, M(0)) = 0 \mid Y(1, M(1)) = 1, M(1) = 1, X = x).
\end{align}

\noindent $\psi(x)$ represents the probability of indirect causation, while $\zeta(x)$ is the probability of direct causation. We use these terms because $\psi(x)$ represents the probability that the negative outcome was induced via the $A \to M \to Y$ pathway, while the term $\zeta(x)$ reflects the probability that the outcome was induced via the $A \to Y$ pathway. The total mediated probability of causation is again simply the sum of these two quantities.

We can also define versions of these estimands that only condition on observed variables. These are again equivalent to the quantities introduced above under the identifying assumptions we outline below.
\begin{align}
\tilde{\psi}(x) &= P(Y(1, M(0)) = 0, Y(0, M(0)) = 0 \mid Y = 1, M = 1, A = 1, X = x) \\
\tilde{\zeta}(x) &= P(Y(1, M(0)) = 1, Y(0, M(0)) = 0 \mid Y = 1, M = 1, A = 1, X = x) 
\end{align}

Which estimand is of primary interest depends on the specific research question, and possibly which pathways are legally or ethically impermissible. For example, if the indirect channel is legally or ethically impermissible, then we may wish to know the probability of indirect causation $\psi(x)$; if the direct channel, then the probability of direct causation $\zeta(x)$. If either pathway are of interest, and we observe a negative mediating event, then the total mediated probability of causation $\delta(x)$. Finally, when any possible pathway through which the exposure might induce the outcome are of interest, then the probability of causation $\tau(x)$ may be the relevant causal quantity. 

\subsection{Other related quantities}\label{ssec:estdisc}

One possible concern is that our proposed probabilities of direct and indirect causation decompose the total mediated probability of causation rather than the probability of causation itself. Mathematically, we could have instead defined quantities similar to $\psi(x)$ and $\zeta(x)$ that do not condition on the potential mediator $M(1)$, and that have the seemingly desirable property of summing to the probability of causation instead of the total mediated probability of causation. Moreover, these quantities would capture information from the term $\delta'(x)$ defined above, which represented an analogous total mediated probability of causation term defined on the $M(1) = 0$ stratum.

We argue that these quantities have little practical interest: once we observe the mediator, we know whether or not the negative mediating event occurred, and it is thus natural to condition on this information. And if we do not observe the post-exposure mediator, we show in Appendix \ref{app:othres} that we cannot generally identify these probabilities, so that the question becomes strictly theoretic. Finally, we may be in a setting where we observe a mediator, but the negative outcome did not occur, so that the term $\delta'(x)$ and an analogous decomposition into direct and indirect pathways may seem desirable. However, if the negative mediator did not occur under exposure, then, under the monotonicity conditions we will require for causal identification, the outcome can only have been induced via a direct pathway -- that is, pathways not through the mediator.\footnote{Intuitively, monotonicity rules out the case where the exposure is protective of the mediator; that is, the mediating event would have occurred absent exposure but did not occur due to the exposure.} No further decomposition of this term is necessary. In short, because most conceivable scenarios where we would wish to estimate probabilities of direct versus indirect causation involve the case where we observe a negative mediating event under exposure, we argue that the total mediated probability of causation and its decomposition are the most relevant quantities to consider in practice. We nevertheless provide expressions for these related quantities and identification results in Appendix \ref{app:othres}.

\subsection{Identification}

We can identify the mediated probabilities of causation in the observed data assuming that for all values of $a, m$:

\begin{assumption}[Y-M consistency]\label{ass:consistency}
\begin{align*}
    A = a, M = m &\implies Y = Y(a, m) \\
    A = a &\implies M = M(a)
\end{align*}
\end{assumption}

\begin{assumption}[Y-M monotonicity]\label{ass:monotonicity}
\begin{align*}
    Y(1, 1) &\ge Y(1, 0) \ge Y(0, 0), \\
    Y(1, 1) &\ge Y(0, 1), \\
    M(1) &\ge M(0)
\end{align*}
\end{assumption}

\begin{assumption}[A-YM ignorability]\label{ass:aym}
\begin{align*}
    A \perp \{Y(1, 1), Y(1, 0), Y(0, 1), Y(0, 0), M(1), M(0)\} \mid X 
\end{align*}
\end{assumption}

\begin{assumption}[Cross-world ignorability]\label{ass:xworld}
\begin{align*}
    \{Y(1, 1), Y(1, 0), Y(0, 1), Y(0, 0)\} \perp \{M(1), M(0)\} \mid X 
\end{align*}
\end{assumption}

\begin{assumption}[A-M positivity]\label{ass:positivity}
\begin{align*}
    P\{\min_{a, m}P(A = a, M = m \mid X) \ge \epsilon \} = 1, \qquad \epsilon > 0
\end{align*}
\end{assumption}

\begin{assumption}[Y positivity]\label{ass:ypositivity}
\begin{align*}
    P\{P(Y = 1 \mid A = 1, M = 1, X) \ge \epsilon \} = 1, \qquad \epsilon > 0
\end{align*}
\end{assumption}

Assumption (\ref{ass:consistency}) precludes interference between individuals; that is, each individual's potential outcomes only depends on their own mediator and exposure status. Assumption (\ref{ass:monotonicity}) implies that the mediator and exposure can only work to induce the outcomes, and the exposure can only work to induce the mediator, with one exception: the mediator may be protective against the outcome in the absence of the exposure. Assumption (\ref{ass:aym}) implies that the exposure is effectively randomized with the potential outcomes and mediators given the covariates. 

Versions of Assumption (\ref{ass:xworld}) are common in the mediation literature to identify natural effects, and this assumption implies that the potential mediators are independent of the potential outcomes given covariates. This is a strong assumption and is unenforceable even in a randomized experiment: we cannot observe the counterfactuals $M(1)$ ($M(0)$) and $Y(0, m)$ ($Y(1, m)$) for the same individual, since we cannot observe the same individual under multiple values of the exposure. Assessing whether violations of this assumption plausibly occur is ultimately application specific, although examples of such scenarios -- which are in general challenging to envision -- have been discussed at length elsewhere (see, e.g., \cite{andrews2021insights}). In any case, cross-world ignorability is routinely invoked in identifying (average) direct and indirect effects.

Assumption (\ref{ass:positivity}) implies that there is some positive probability of observing all exposure and mediator combinations for any covariate value; finally, assumption (\ref{ass:ypositivity}) implies that there is a positive probability of experiencing the outcome for any covariate value under exposure and the mediator.

\begin{remark}
We can weaken some of the above assumptions to identify $\psi(x)$ or $\delta(x)$ alone; moreover, these assumptions are somewhat stronger than strictly necessary to identify these three estimands. We provide slightly weaker sets of identifying conditions for each individual estimand in Appendix \ref{app:identification}.
\end{remark}

\begin{remark}
Substantially weaker assumptions may identify natural effects. For example, assumption \ref{ass:monotonicity} is not required; assumption \ref{ass:positivity} is stronger than needed (see, e.g., \cite{vansteelandt2017interventional}); and no analogue of assumption \ref{ass:ypositivity} is needed.\footnote{This is required here because the identified estimands involve risk-ratios with this quantity in the denominator. Assumption \ref{ass:ypositivity} ensures these quantities are well-defined for any chosen covariate value.} Finally, in place of assumption \ref{ass:xworld}, one typically assumes,

\begin{align}
\label{eqn:standardym} Y(a, m) &\perp M \mid \{X, A = a\}, \qquad a = 0, 1 \\
\label{eqn:standardxw} Y(a, m) &\perp M(a') \mid X, \qquad m, a, a' = 0,1; a \ne a'.
\end{align}

\noindent Assumptions (\ref{ass:consistency}) and (\ref{ass:aym}) imply equation \ref{eqn:standardym}, and assumption \ref{ass:xworld} alone implies equation (\ref{eqn:standardxw}); however, the reverse is not true.
\end{remark}

Theorem \ref{theorem1} provides identifying expressions for $\delta(x)$ and $\psi(x)$ in terms of the observed data distribution. We define the observed data quantities $\mu_{am}(x) = P(Y = 1 \mid A = a, M = m, X = x)$, and $\gamma_a(x) = P(M = 1 \mid A = a, X = x)$.

\begin{theorem}[Identification]\label{theorem1}
Under assumptions (\ref{ass:consistency})-(\ref{ass:ypositivity}), $\psi(x)$, $\delta(x)$, and $\zeta(x)$ are identified in the observed data as,

\begin{align}\label{ident2}
\psi(x) &= \left[1 - \frac{\mu_{10}(x)}{\mu_{11}(x)}\right]\left[1 - \frac{\gamma_0(x)}{\gamma_1(x)}\right], \\
\delta(x) &= \left[1-\frac{\mu_{00}(x)}{\mu_{11}(x)}\right]\left[1 - \frac{\gamma_0(x)}{\gamma_1(x)}\right] + \left[1-\frac{\mu_{01}(x)}{\mu_{11}(x)}\right]\left[\frac{\gamma_0(x)}{\gamma_1(x)}\right], \\
\zeta(x) &= \delta(x) - \psi(x).
\end{align}
\end{theorem}

The identifying expressions in Theorem (\ref{theorem1}) have intuitive interpretations. First, consider the expression for $\psi(x)$. Under our causal assumptions, this expression is equivalent to \footnote{This identity is implied by the Proof in the Appendix.}

\begin{align*}
&P(Y(1, 0) = 0 \mid Y(1, 1) = 1, M(1) = 1, X = x)P(M(0) = 0 \mid M(1) = 1, X = x).
\end{align*}

\noindent The first quantity in this expression represents the total mediated probability of causation with respect to a ``controlled indirect effect.'' In contrast to natural effects, which allow the mediator to assume its natural value under the relevant exposure condition, controlled (in)direct effects conceive of deterministically setting the exposure and mediator values. This functional therefore represents the probability that a negative outcome was induced by changing the mediator from $0$ to $1$ in the presence of an exposure among the stratum where $Y(1, 1) = 1$ and $M(1) = 1$. The second term is simply the probability of causation with respect to the mediator; that is, it is the probability of causation where we allow the mediator to play the role of the outcome. The identifying expression for $\psi(x)$ is therefore equivalent to the probability that the exposure induced the mediator times the probability that the inducing the mediator caused the outcome in the presence of the exposure (i.e. changing $Y(1, 0)$ to $Y(1, 1)$).

The identifying expression for $\delta(x)$ has a slightly more complex, though still intuitive, explanation. In this case, we can show that this expression is equivalent to,
\small
\begin{align*}
&P(Y(0, 0) = 0 \mid Y(1, 1) = 1, M(1) = 1, M(0) = 0, X = x)P(M(0) = 0 \mid M(1) = 1, X = x) \\
&+P(Y(0, 1) = 0 \mid Y(1, 1) = 1, M(1) = 1, M(0) = 1, X = x)P(M(0) = 1 \mid M(1) = 1, X = x).
\end{align*}
\normalsize

\noindent In other words, the total mediated probability of causation is equal to the probability that the exposure induced the mediator and either the mediator or the exposure induced the outcome; plus the probability that the exposure did not induce the mediator, but did induce the outcome in the presence of the mediator. These represent all possible pathways from which $A$ might affect $Y$ under the conditioning events. Finally, to arrive at the probability of indirect causation, we simply subtract from the expression for the probability of direct causation from the expression for the total mediated probability of causation. Intuitively, this removes the particular case where the exposure induced the mediator and the mediator induced the outcome, leaving only direct pathways from the exposure to the outcome behind; that is, cases where the exposure alone induced the outcome regardless of whether the exposure induced the mediator.

As a separate and perhaps useful point, the identification result for $\psi(x)$ implies that a valid upper bound on the probability of indirect causation is given by the probability of causation with respect to the mediator. This follows because in the most extreme case the probability of causation via a controlled indirect effect is one: in this case it remains only to show that the mediator was induced by the exposure. Conveniently, identification of this bound only requires assuming assumptions (\ref{eqn:pc1})-(\ref{eqn:pc5}) using $M$ and $M(a)$ in place of $Y$ and $Y(a)$. Thus, in practice, if the legal standard for some case depended on the probability of indirect causation being less than some threshold, it would suffice to show that the probability of causation with respect to the mediator is less than this same threshold. Because the probability of causation requires far weaker identifying assumptions than the probability of indirect causation, this bound may be useful for some real-world settings.

\subsection{Graphical intuition}
To illustrate the relationships between the causal estimands we have discussed -- the average treatment effect, the probability of causation, and the mediated probabilities of causation -- we simulate data and display the implied estimands.

 \begin{figure}
 \begin{center}
     \includegraphics[scale=0.4]{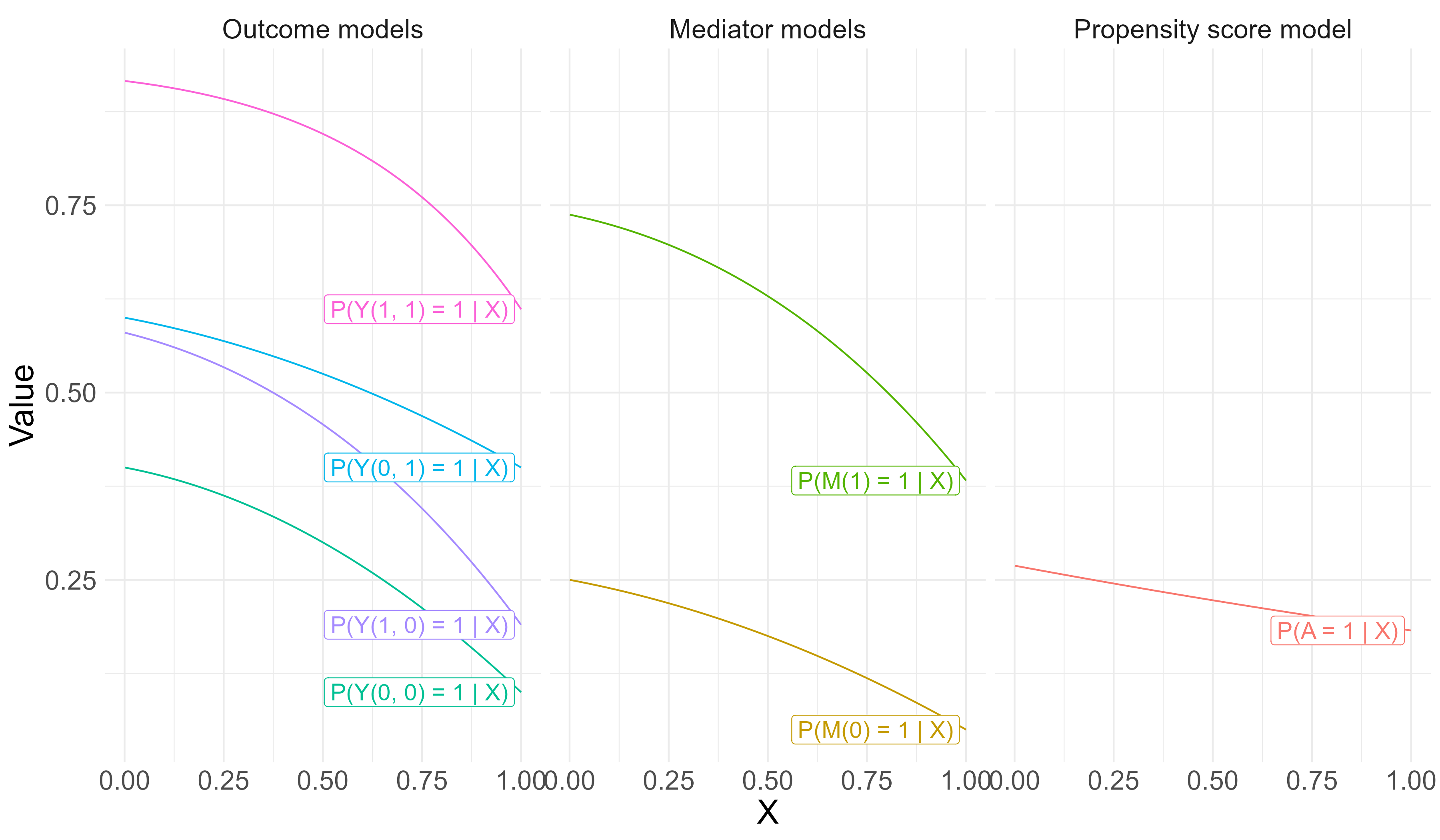}
 \end{center}
 \caption{Nuisance components}\label{fig:0}
 \end{figure}

We first draw a single continuous covariate $X \stackrel{iid}\sim \mathcal{U}(0, 1)$. We then generate propensity score, mediator, and outcome models as polynomial functions of this variable, as illustrated in Figure \ref{fig:0} (the precise specifications are available in Appendix \ref{app:specs}). 

The simulated functions above suggest that the probability of the outcome, mediator, and exposure are all decreasing in the covariate value. Moreover, for any fixed value of the covariate, the simulated functions give the exposed group a higher probability of the mediator than the corresponding non-exposed group ($P(M(1) = 1 \mid x) > P(M(0) = 1 \mid x)$ for all $x$). Moreover, for a fixed value of the mediator and covariate, the probability of the outcome for the exposed group is greater compared to the unexposed group  (i.e. $P(Y(1, m) = 1 \mid x) > P(Y(0, m) = 1 \mid x)$ for all $m, x$). Finally, we ensure that all monotonicity assumptions hold when making Bernoulli draws of the potential outcomes and mediators from these distributions.

 \begin{figure}
 \begin{center}
     \includegraphics[scale=0.4]{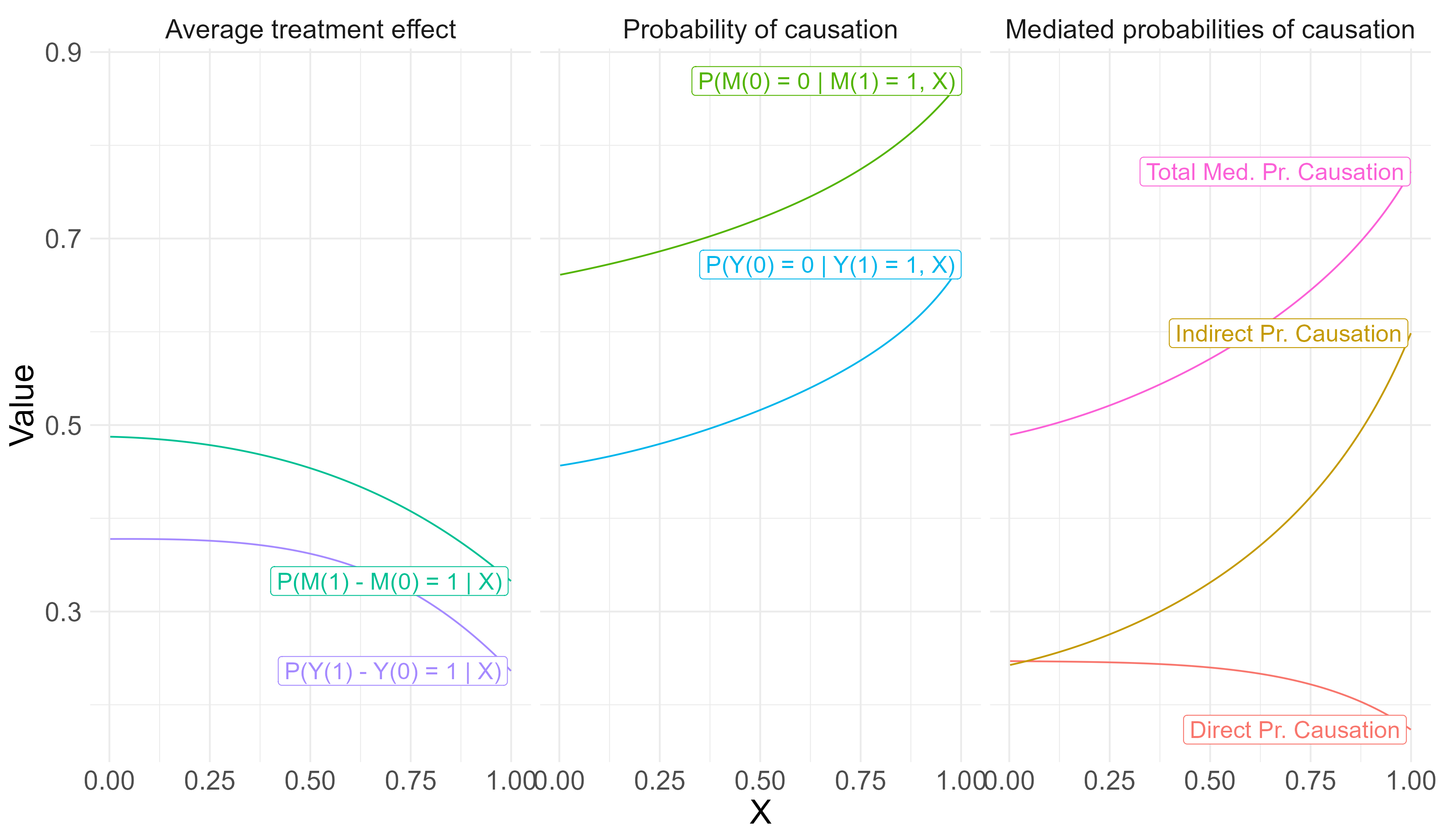}
 \end{center}
 \caption{Causal estimands}\label{fig:1}
 \end{figure}

Figure \ref{fig:1} displays the implied causal estimands in our simulated dataset. The first panel depicts the (conditional) average treatment effects. The plot shows that the effects on both the mediators and the outcome decrease as the covariate value increases. However, these functions do not answer the question: given that an individual would experience the negative outcome were they exposed, would they still experience the negative outcome were they not exposed? Or the question, given that an individual would experience a negative mediator were they exposed, would they not experience the negative mediator were they not exposed? These are the probabilities of causation with respect to the outcome and the mediator, respectively, and are displayed in the second panel. In contrast to the average treatment effects, the values of both quantities increase with the covariate value.

Finally, the third panel depicts the mediated probabilities of causation: these estimands condition on the subset of individuals who would experience the negative mediator and outcome were they exposed. The pink line represents the total mediated probability of causation, while the yellow and red lines represent the probabilities of indirect and direct causation, respectively. Throughout most of the covariate distribution, the probability of indirect causation is higher than the probability of direct causation, and for individuals with values of $x$ approximately greater than $0.885$, the probability of indirect causation is greater than 0.5. For these individuals, it is more likely than not that they would not experience the outcome absent exposure, due to the fact that they would not experience the negative mediating event absent exposure.

\section{Estimation}\label{sec:estimation}

We propose estimating projections of the mediated probabilities of causation onto parametric working models, following \cite{cuellar2020non}. In other words, we do not propose estimating the mediated probabilities of causation directly, but instead target a summary measure. Figure \ref{fig:2} illustrates this idea, displaying the true total mediated and probability of indirect of causation as a function of $x$, denoted in the darker lines, versus their best fitting linear approximation based on minimizing the squared error, denoted in the lighter shades of the same color. Notably, the projections are almost never exactly equal to the true quantities because the true functions are not linear.\footnote{In practice, we could define a more complex projection that better fits the true data generating process, though we do not do so to illustrate the difference between these quantities.} Nevertheless, these quantities are generally close, illustrating that projections may be useful targets of inference. 

 \begin{figure}
 \begin{center}
     \includegraphics[scale=0.4]{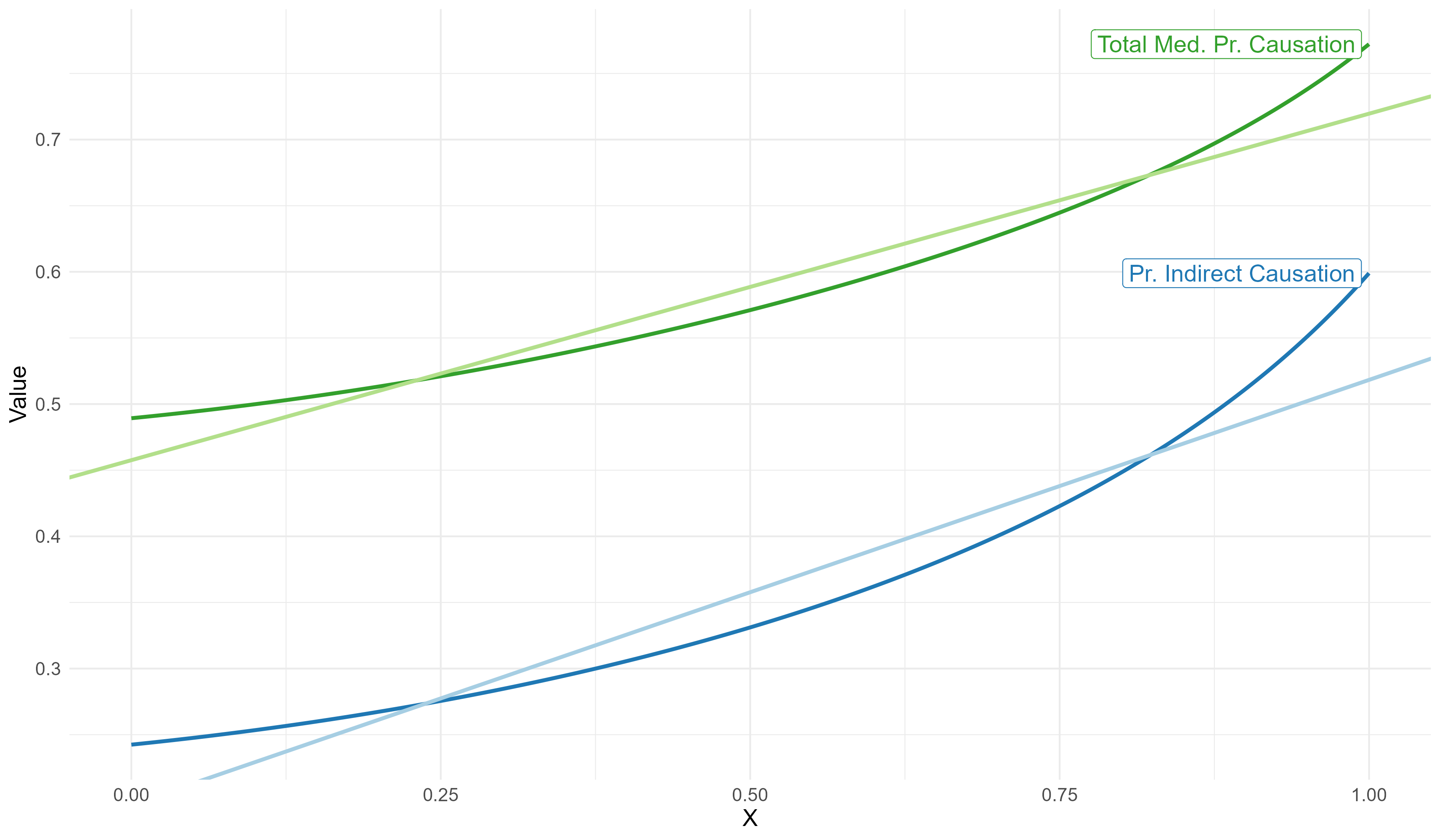}
 \end{center}
 \caption{Projections}\label{fig:2}
 \end{figure}

Targeting projections is also desirable from a statistical perspective because, when the covariates are continuous, these estimands cannot generally be estimated at root-n rates without making strong -- and unrealistic -- parametric modeling assumptions. By contrast, we can obtain root-n consistent and asymptotically normal estimates of projections of these parameters under relatively weak conditions. We briefly review this approach and elaborate on the properties of these estimators below, using a projection of $\psi(x)$ to illustrate. 

To be precise, we define a projection of $\psi(x)$ at a point $x$ as $g(x; \beta)$ for some function $g$ parameterized by some finite-dimensional parameter vector $\beta$, where $\beta$ is defined as the solution to equation (\ref{eqn:projection})\footnote{While we define the projection with respect to the $L_2$ loss, we could generalize this formula to minimize other loss functions instead.},

\begin{align}\label{eqn:projection}
    \beta = \arg\min_{\tilde{\beta}}\mathbb{E}[w(X)\{\psi(X) - g(X; \tilde{\beta})\}^2].
\end{align}

\noindent The solution to (\ref{eqn:projection}) depends on pre-specified weights $w(X)$ that differentially weight the covariate space. In the simplest case, these weights may be uniform, so that all portions of the covariate space contribute equally to the projection; however, other choices of weights are possible. Moreover, we need not assume $\psi(x) = g(x; \beta)$ for any $x$ for the projection defined by $g(x; \beta)$ to be a valid and meaningful target of statistical inference. This fact generally holds when using parametric models to estimate unknown and possibly complex data generating mechanisms. Indeed, one can often interpret model parameters as defining a projection of some true process under no modeling assumptions \cite{angrist2009mostly, buja2019models}.\footnote{In some cases we may wish to define a projection with respect to some subset of the covariates, for example, when the covariates are high-dimensional. In this case we can simply define the function $g$ as a function of this subset. The estimation strategy outlined below proceeds analogously.}

Notice, however, that in cases where the projections of all mediated probabilities of causation are desired, no guarantee exists that the projections of the probabilities of direct and indirect causation sum to the projection of the total mediated probability of causation, except in special cases (for example, when all models are linear in the parameters of identical covariate transformations and an identical weighting scheme is used across all models). If this property is desirable, one can simply estimate two of the three projections and take the relevant sum or difference to obtain the third. While this final quantity is not strictly speaking the projection, it is still a well-defined functional that may be of interest.\footnote{That said, if two of the projections sum (or difference) to something far from the projection of the third, this may suggest that one or more of the projections is badly misspecified for the true function.}

Regardless, to estimate a projection, we follow \cite{cuellar2020non} and propose a ``doubly-robust'' estimation approach. Specifically, we propose using the estimator $\hat{\beta}$ that satisfies the moment condition,

\begin{align}\label{eqn:estimator}
\mathbb{P}_n\left[\frac{\partial g(X; \hat{\beta})}{\partial \beta} w(X)\left(\varphi(O; \hat{\eta})  - g(X; \hat{\beta})\right)\right] = 0,
\end{align}

\noindent where $\varphi(O; \hat{\eta})$ uses estimated ``nuisance'' functions $\hat{\eta}$ of $\eta = [\mu_{11}, \mu_{10}, P(A, M \mid X)]$\footnote{Note that we can recover estimates of $(P(A = 1 \mid X), \gamma_1(X), \gamma_0(X))$ from estimates of $P(A, M \mid X)$. Alternatively, we could separately estimate these functions to obtain estimates of $P(A, M \mid X)$.} plugged into the following expression,

\begin{align*}
\varphi(O; \eta) = \varphi_3(O; \eta) - \varphi_1(O; \eta) - \varphi_2(O; \eta) + \psi(X; \eta),
\end{align*}

\noindent where

\begin{align*}
\varphi_1(O; \eta) &= \frac{1}{\mu_{11}(X)}\left[\frac{A(1-M)\{Y - \mu_{10}(X)\}}{P(A = 1, M = 0 \mid X)} - \frac{\mu_{10}(X)}{\mu_{11}(X)}\frac{AM\{Y - \mu_{11}(X)\}}{P(A = 1, M = 1 \mid X)}\right], \\
\varphi_2(O; \eta) &= \frac{1}{\gamma_1(X)}\left[\frac{(1-A)(M - \gamma_0(X))}{P(A = 0 \mid X)} - \frac{\gamma_0(X)}{\gamma_1(X)}\frac{A(M - \gamma_1(X))}{P(A = 1 \mid X)}\right],\\
\varphi_3(O; \eta) &= \varphi_1(O; \eta)\left[\frac{\gamma_0(X)}{\gamma_1(X)}\right] + \varphi_2(O; \eta)\left[\frac{\mu_{10}(X)}{\mu_{11}(X)}\right].
\end{align*}

Theorem \ref{theorem3} in Appendix \ref{app:othest} formalizes conditions where this estimator will yield root-n consistent and asymptotically normal estimates of $\beta$ and therefore $g(x; \beta)$. 

A key benefit of this ``doubly-robust'' estimation strategy versus other possible methods is that it allows for non-parametric estimation of the nuisance functions while still obtaining parametric rates of convergence of the estimated projection parameter under relatively mild conditions. At a high level, this is possible because the error of doubly-robust estimators is a function of the product of errors among the estimated nuisance functions. As a result, the overall error of the estimator converges to zero more quickly than each nuisance component. As long as certain error products between nuisance components are estimated quickly enough (to be precise, are $o_p(n^{-1/2})$ in the $L_2$ norm), then the nuisance estimation can be ignored asymptotically. In other words, the asymptotics of these estimators is as if we regressed $\varphi(O; \eta)$ onto $g(X; \beta)$ directly, instead of the estimated function $\varphi(O; \hat{\eta})$.

We refer to \cite{cuellar2020non} for a more thorough discussion of using a non-parametric projection based estimation methods and \cite{kennedy2021semiparametric} for a more thorough review of doubly-robust estimators. We also provide the analogous doubly-robust estimators for the projections of $\delta(x)$ and $\zeta(x)$ in Appendix \ref{app:othest}.

\section{Simulations}\label{sec:simulations}

We conduct a simulation study to verify the expected theoretic performance of our proposed estimators. In our first simulation study we examine the bias, RMSE, and coverage rates for the proposed projection parameter while simulating nuisance estimation error on the true data generating functions. In particular, we draw 1,000 samples of size $n = 1,000$ and add gaussian noise to each nuisance parameter $\eta_j$ in the following way:

\begin{align*}
    \hat{\eta}_j(x) = \text{expit}(\text{logit}(\eta_j(x)) + \mathcal{N}(C_{1j}n^{-\alpha}, C_{2j}n^{-2\alpha}))
\end{align*}

\noindent for constants $C_1$ and $C_2$ that depend on the nuisance function and $\alpha$ that defines the convergence rate. In other words, we simulate nuisance parameters estimated at $\mathcal{O}_p(n^{-\alpha})$ rates in the $L_2$ norm. We test $\alpha = 0.3, 0.1$, where the first parameterization is fast enough to satisfy the rate conditions of Theorem \ref{theorem3}, but the second does not.\footnote{This theorem states that a sufficient condition for root-n consistency and asymptotic normality of this parameter estimate is that each nuisance function be estimated at $o_p(n^{-1/4})$ rates.} We then estimate a linear projection of each mediated probability of causation (total, indirect, and direct) at the value of $x = 0.75$, as illustrated in Figure \ref{fig:1} above, compared to the ``true'' projection parameter, which we estimate by simulating a population of ten million and estimating the projection using the same model. Table \ref{tab:1} displays the results of this first simulation, illustrating minimal bias for all estimands and showing that we are able to obtain approximately nominal coverage rates for the projection parameters when $\alpha = 0.3$, but not when $\alpha = 0.1$, as expected. Code to implement these simulations is available at: https://github.com/mrubinst757/pcmediation.

\begin{table}
\caption{Simulation results}\label{tab:1}
\centering
\begin{tabular}[t]{cccccccc}
\toprule
\multicolumn{2}{c}{ } & \multicolumn{3}{c}{$\alpha = 0.3$} & \multicolumn{3}{c}{$\alpha = 0.1$} \\
\cmidrule(l{3pt}r{3pt}){3-5} \cmidrule(l{3pt}r{3pt}){6-8}
Estimand & Truth & Bias & RMSE & Coverage & Bias & RMSE & Coverage\\
\midrule
Indirect & 0.44 & -0.01 & 0.05 & 0.96 & 0.23 & 0.30 & 0.80\\
Direct & 0.22 & 0.02 & 0.05 & 0.94 & -0.13 & 0.23 & 0.92\\
Total & 0.65 & 0.01 & 0.04 & 0.94 & 0.10 & 0.15 & 0.89\\
\bottomrule
\end{tabular}
\end{table}

\section{Conclusion}\label{sec:discussion}

We have introduced a class of causal estimands that merge ideas from the literature on mediation and the probability of causation that we call the ``mediated probabilities of causation.'' We highlighted three quantities in particular: the total mediated probability of causation, and the probabilities of direct and indirect causation. We then provided a set of assumptions sufficient to identify these estimands in observed data. While these assumptions are quite strong, we also show that a far weaker set of assumptions suffices to identify upper bounds on the probability of indirect causation. Finally, we propose estimating these quantities using a doubly-robust projection based approach, following \cite{cuellar2020non}. In a set of simulation studies we verify that these estimators can obtain root-n consistent and asymptotically normal estimates of the projection parameters while allowing for relatively slow rates of convergence on the nuisance estimates. 

Although we do not apply our methodology to any specific empirical example, we hope that it is useful to formally explicate these quantities. For example, the formalism that we bring to these quantities may aid legal discussions, including the Harvard admissions case, and clarify the strong assumptions required to identify and estimate these parameters in observed data.\footnote{https://www.supremecourt.gov/opinions/22pdf/20-1199_hgdj.pdf} Moreover, this formalism may also be useful in medical contexts. For example, if medical providers or other stakeholders hypothesize that a vaccinated patient's negative outcome was brought about by a negative mediating event -- for example, an adverse reaction to vaccination -- then the probability of indirect causation may be of substantive interest.

This paper also has technical limitations. For example, we only consider the single mediator setting, and restrict that the mediators, exposure, and outcome are all binary. Applied settings frequently have multiple possibly continuous mediators with unknown causal orderings. Formalizing analogous estimands in these more complicated -- and realistic -- settings would be a useful area for future research. Additionally, our identification results require strong unverifiable assumptions, including cross-world ignorability. Whether it is possible to weaken these assumptions, or bound these estimands under some assumed sensitivity parameters, would be interesting to investigate. Finally, our proposed estimation method targets a projection of the causal estimands, rather than the causal estimands themselves. We could instead use DR-learner \cite{kennedy2022} to estimate the true mediated probabilities of causation.\footnote{At a high-level, this simply involves replacing the parametric projection with a fully non-parametric model and regressing the estimated quantity $\varphi(O; \hat{\eta})$ onto this quantity.} While this approach will not generally yield root-n consistent estimates, a full articulation of this method applied to this setting may be worthwhile. \\ [2ex]

\noindent\textbf{Acknowledgments:} The authors would like to thank two anonymous reviewers and the Associate Editor for helpful comments, questions, and suggestions that improved the quality of this manuscript. \hfill \break

\noindent\textbf{Funding information:} DM was partially supported by the National Institutes of Health under award number K25ES034064 from NIEHS. \hfill \break

\noindent\textbf{Conflict of interest:} Authors state no conflicts of interest. \hfill \break

\noindent\textbf{Data availability statement:} Code to replicate the simulation results is available at: \url{https://github.com/mrubinst757/pcmediation}

\bibliographystyle{vancouver} 
\bibliography{research.bib}       

\newpage 

\appendix

\section{Additional theoretic results}\label{app:othest}

Proposition \ref{proposition2} formalizes the proposed influence function-based estimator discussed in Section \ref{sec:estimation}. 

\begin{proposition}\label{proposition2}

Consider the moment condition for a fixed $\beta^\star$ suggested by equation (\ref{eqn:projection}):

\begin{align*}
    \Psi(\beta^\star) = \mathbb{E}[w(X)(\psi(X) - g(X; \beta^\star))] = 0
\end{align*}

\noindent where $\psi(x)$ is given by equation (\ref{eqn:ident1}). Under a non-parametric model, the uncentered efficient influence curve for the moment condition $\Psi(\beta^\star)$ at any fixed $\beta^\star$ is given by

\begin{align}
\phi(Z; \beta^\star, \eta) &= \frac{\partial g(X; \beta^\star)}{\partial \beta}w(X)\left(\varphi(O; \eta)  - g(X; \beta^\star)\right)
\end{align}

\noindent where 

\begin{align*}
\varphi(O; \eta) = \varphi_3(O; \eta) - \varphi_1(O; \eta) - \varphi_2(O; \eta) + \psi(X; \eta) 
\end{align*}

\noindent and

\begin{align*}
\varphi_1(O; \eta) &= \frac{1}{\mu_{11}(X)}\left[\frac{A(1-M)\{Y - \mu_{10}(X)\}}{P(A = 1, M = 0 \mid X)} - \frac{\mu_{10}(X)}{\mu_{11}(X)}\frac{AM\{Y - \mu_{11}(X)\}}{P(A = 1, M = 1 \mid X)}\right] \\
\varphi_2(O; \eta) &= \frac{1}{\gamma_1(X)}\left[\frac{(1-A)(M - \gamma_0(X))}{P(A = 0 \mid X)} - \frac{\gamma_0(X)}{\gamma_1(X)}\frac{A(M - \gamma_1(X))}{P(A = 1 \mid X)}\right]\\
\varphi_3(O; \eta) &= \varphi_1(O; \eta)\left[\frac{\gamma_0(X)}{\gamma_1(X)}\right] + \varphi_2(O; \eta)\left[\frac{\mu_{10}(X)}{\mu_{11}(X)}\right].
\end{align*}

\end{proposition}

Proposition \ref{proposition3} provides an analogous estimator for the projection of $\delta(x)$. 

\begin{proposition}\label{proposition3}

Consider the moment condition for a fixed $\alpha^\star$ suggested by equation (\ref{eqn:projection}), where we assume the user-specified projection of $\delta(x)$ (see equation (\ref{ident2}), is given by the function $h(x; \alpha^\star)$:

\begin{align*}
    \Psi_{tpc}(\alpha^\star) = \mathbb{E}[w(X)(\delta(X) - h(X; \alpha^\star))] = 0
\end{align*}

\noindent Under a non-parametric model, the uncentered efficient influence curve for the moment condition $\Psi_{tpc}(\alpha^\star)$ at any fixed $\alpha^\star$ is given by

\begin{align}
\phi_{tpc}(O; \alpha^\star, \eta) &= \frac{\partial h(X; \alpha^\star)}{\partial \alpha}w(X)\left(\varphi_{tpc}(O; \eta)  - h(X; \alpha^\star)\right)
\end{align}

\noindent where 

\begin{align*}
\varphi_{tpc}(O; \eta) = \varphi_{tpc, 4}(O; \eta) - \varphi_{tpc, 5}(O; \eta) - \varphi_{tpc, 1}(O; \eta) + \delta(X; \eta) 
\end{align*}

\noindent and

\begin{align*}
\varphi_{tpc, 1}(O; \eta) &= \frac{1}{\mu_{11}(X)}\left[\frac{(1-A)(1-M)\{Y - \mu_{00}(X)\}}{P(A = 0, M = 0 \mid X)} - \frac{\mu_{00}(X)}{\mu_{11}(X)}\frac{AM\{Y - \mu_{11}(X)\}}{P(A = 1, M = 1 \mid X)}\right] \\
\varphi_{tpc, 2}(O; \eta) &= \frac{1}{\mu_{11}(X)}\left[\frac{(1-A)M\{Y - \mu_{01}(X)\}}{P(A = 0, M = 1 \mid X)} - \frac{\mu_{01}(X)}{\mu_{11}(X)}\frac{AM\{Y - \mu_{11}(X)\}}{P(A = 1, M = 1 \mid X)}\right] \\
\varphi_{tpc, 3}(O; \eta) &= \frac{1}{\gamma_1(X)}\left[\frac{(1-A)(M - \gamma_0(X))}{P(A = 0 \mid X)} - \frac{\gamma_0(X)}{\gamma_1(X)}\frac{A(M - \gamma_1(X))}{P(A = 1 \mid X)}\right] \\
\varphi_{tpc, 4}(O; \eta) &= \varphi_{tpc, 1}(O; \eta)\left[\frac{\gamma_0(X)}{\gamma_1(X)}\right] + \varphi_{tpc, 3}(O; \eta)\left[\frac{\mu_{00}(X)}{\mu_{11}(X)}\right] \\
\varphi_{tpc, 5}(O; \eta) &= \varphi_{tpc, 3}(O; \eta)\left[\frac{\gamma_0(X)}{\gamma_1(X)}\right] + \varphi_{tpc, 3}(O; \eta)\left[\frac{\mu_{01}(X)}{\mu_{11}(X)}\right]
\end{align*}

\noindent where $\eta = [\mu_{11}, \mu_{00}, \mu_{01}, P(A, M \mid X)]$.
\end{proposition}

\begin{remark}
    To obtain an analogous result for the quantity $\zeta(x)$, one simply need note that $\zeta(x) = \delta(x) - \psi(x)$: since the influence function of the difference of two functionals is the difference of the influence functions, the result immediately follows that the influence function for the corresponding moment condition is given by:

    \begin{align}
        \phi_{pnde}(O; \rho^\star, \eta) &= \frac{\partial m(X; \rho^\star)}{\partial \rho}w(X)\left(\varphi_{tpc}(O; \eta) - \varphi(O; \eta)  - m(X; \rho^\star)\right)
    \end{align}
    \noindent where $m(x; \rho)$ defines the user-specified projection of $\zeta(x)$.
\end{remark}

Theorem \ref{theorem3} establishes the conditions where the estimator proposed in equation (\ref{eqn:estimator}) is root-n consistent for $\psi$ and asymptotically normal.

\begin{theorem}\label{theorem3}
Consider the moment condition $\mathbb{E}[\phi(O; \beta_0, \eta_0)] = 0$ evaluated at the true parameters $(\beta_0, \eta_0)$. Now consider the estimator $\hat{\beta}$ that satisfies $\mathbb{P}_n[\phi(O; \hat{\beta}, \hat{\eta})] = 0$, where $\hat{\eta}$ is estimated on an independent sample. Assume that:
\begin{itemize}
    \item The function class $\{\phi(O; \beta, \eta): \beta \in \mathbb{R}^p\}$ is Donsker in $\beta$ for any fixed $\eta$
    \item $\|\phi(O; \hat{\beta}, \hat{\eta}) - \phi(O; \beta_0, \eta_0)\| = o_p(1)$
    \item The map $\beta \to \mathbb{P}[\phi(O; \beta, \eta)]$ is differentiable at $\beta_0$ uniformly in $\eta$, with non-singular derivative matrix $\frac{\partial}{\partial \beta}\mathbb{P}\{\phi(O; \beta, \eta)\}\mid_{\beta = \beta_0} = M(\beta_0, \eta)$, where $$M(\beta_0, \hat{\eta}) \overset{p}{\to} M(\beta_0, \eta_0)$$
    \item $\|\hat{\eta} - \eta\| = o_p(n^{-1/4})$
\end{itemize}

\noindent Then the proposed estimator attains the non-parametric efficiency bound and is asymptotically normal with

\begin{align}
    \sqrt{n}(\hat{\beta} - \beta_0) \to^d \mathcal{N}(0, M^{-1}\mathbb{E}[\phi\phi^\top]M^{-\top}) 
\end{align}

\noindent and, therefore, for any fixed value of $X = x$,

\begin{align}\label{eqn:avarproj}
    \sqrt{n}(g(x; \hat{\beta}) - g(x; \beta_0)) \to^d \mathcal{N}\left(0, \left(\frac{\partial g(x; \beta_0)}{\partial \beta}\right)^\top M^{-1}\mathbb{E}[\phi\phi^\top]M^{-\top}\frac{\partial g(x; \beta_0)}{\partial \beta}\right).
\end{align}
\end{theorem}

\begin{remark}
Theorem \ref{theorem3} assumes that $\|\hat{\eta} - \eta\| = o_p(n^{-1/4})$. This condition is sufficient, but not necessary, for this result to hold. The key requirement is that the error of the influence function based estimator is a function of the product of errors in the nuisance estimates. As long as these error products are all $o_p(n^{-1/2})$, this result will hold.
\end{remark}

\begin{remark}
Analogous results to Theorem \ref{theorem3} for the projections of $\delta(x)$ and $\zeta(x)$ can easily be derived: the difference comes in the precise form of the second-order remainder terms, which we have omitted here, that are implicitly assumed $o_p(n^{-1/2})$ under the assumption that $\|\hat{\eta} - \eta\| = o_p(n^{-1/4})$. Additionally, further notice that we have slightly abused notation to let $\eta$ to represent any of the nuisance components in the corresponding influence functions; however, the specific components differ slightly depending on the estimand (e.g. $\delta(x)$ depends on $\mu_{01}(x)$ while $\psi(x)$ does not).
\end{remark}

\newpage

\section{Proofs}\label{app:proofs}

We divide the proofs into two subsections: first, proofs concerning identification. Second, proofs concerning estimation. For the proofs on identification, we first outline slightly weaker sets of assumptions (compared to those outlined in Section \ref{sec:identification}) that are sufficient to identify each estimand.

\subsection{Identification}\label{app:identification}

\subsubsection{Probability of causation}

\begin{proof}[Proof of proposition (\ref{prop:1})]
\begin{align*}
    \tau(x) &= P(Y(0) = 0 \mid Y(1) = 1, X = x) \\
    &= P(Y(0) = 0 \mid Y(1) = 1, A = 1, X = x) \\
    &= P(Y(0) = 0 \mid Y = 1, A = 1, X = x) \\
    &= \frac{P(Y - Y(0) = 1 \mid A = 1, X = x)}{P(Y = 1 \mid A = 1, X = x)} \\
    &= 1 - \frac{P(Y(0) = 1 \mid A = 1, X = x)}{P(Y = 1 \mid A = 1, X = x)} \\
    &= 1 - \frac{P(Y = 1 \mid A = 0, X = x)}{P(Y = 1 \mid A = 1, X = x)}
\end{align*}

\noindent where the first line follows by definition, the second by A-Y ignorability, the third by Y consistency, the fourth by Bayes' rule, the fifth by Y monotonicity, and the final equality by A-Y ignorability and Y consistency. 
\end{proof}

\begin{remark}
The equivalence of $\tau(x)$ and $\tilde{\tau}(x)$ is implied in the proof, following from the first three equalities.
\end{remark}

\subsubsection{Total mediated probability of causation}

We invoke assumptions (\ref{ass:consistency}) and (\ref{ass:ypositivity}). We further assume:

\begin{assumptionp}{\ref*{ass:monotonicity}$'$}[Y-M monotonicity]\label{ass:monotonicity1}
\begin{align*}
M(1) &\ge M(0) \\
Y(1, 1) &\ge Y(0, m), \qquad m = 0,1
\end{align*}
\end{assumptionp}

\begin{assumptionp}{\ref*{ass:aym}$'$}[A-YM ignorability]
\begin{align*}
A &\perp \{Y(1, 1), Y(0, m), M(1), M(0)\} \mid X, \qquad m = 0,1 \\
\end{align*}
\end{assumptionp}

\begin{assumptionp}{\ref*{ass:xworld}$'$}[Cross-world ignorability]
\begin{align*}
\{M(0), M(1)\} \perp Y(0, m) \mid X, \qquad m = 0,1 \\
\end{align*}
\end{assumptionp}

\begin{assumptionp}{\ref*{ass:positivity}$'$}[A-M positivity]\label{ass:positivity1}
\begin{align*}
P\{P(A = 0, M = m \mid X) &\ge \epsilon\} = 1,\qquad \epsilon > 0, m = 0, 1 \\
P\{P(A = 1, M = 1 \mid X) &\ge \epsilon\} = 1
\end{align*}
\end{assumptionp}

\begin{proof}[Proof of total mediated probability of causation identification]
We invoke assumptions (\ref{ass:consistency}) and (\ref{ass:ypositivity}), as well as assumptions (\ref{ass:monotonicity1})-(\ref{ass:positivity1}) above.

\begin{align}
\nonumber&P(Y(0, M(0)) = 0 \mid Y(1, M(1)) = 1, M(1) = 1, x) \\
&= \frac{P(Y(0, M(0)) = 0, Y(1, 1) = 1 \mid M(1) = 1, x)}{P(Y(1, 1) = 1 \mid M(1) = 1, x)} \\
\nonumber &= 1 - \frac{P(Y(0, M(0)) = 1 \mid M(1) = 1, x)}{P(Y(1, 1) = 1 \mid M(1) = 1, x)} \\
\label{eqn:start1}&= 1 - \frac{P(Y(0, M(0)) = 1 \mid M(1) = 1, x)}{P(Y = 1 \mid A = 1, M = 1, x)} 
\end{align}
where the first equality holds by Bayes' rule and consistency, the second by Y-M monotonicity, and the third by A-YM ignorability and consistency. By the law of iterated expectation and Y-M consistency we further obtain:

\begin{align*}
&P(Y(0, M(0)) = 1 \mid M(1) = 1, x) \\
&= \sum_m P(Y(0, m) = 1 \mid M(0) = m, M(1) = 1, x)P(M(0) = m \mid M(1) = 1, x) \\
&= \sum_m P(Y(0, m) = 1 \mid M(0) = m, M(1) = 1, x)w_m(x)
\end{align*}
By A-YM ignorability, cross-world ignorability, and consistency we further obtain that:

\begin{align*}
&P(Y(0, m) = 1 \mid M(0) = m, M(1) = 1, x)\\
&= P(Y(0, m) = 1 \mid x) \\
&= P(Y = 1 \mid A = 0, M = m, x)
\end{align*}
Finally, Proposition \ref{prop:1} and the corresponding assumptions imply that:
\begin{align*}
w_0(x) = 1 - \frac{P(M = 1 \mid A = 0, x)}{P(M = 1 \mid A = 1, x)}
\end{align*}
noting that $w_1(x) = 1 - w_0(x)$. Combining the expressions yields the result. 
\end{proof}

\begin{remark}
The equivalence of $\tilde{\delta}(x)$ and $\delta(x)$, and therefore the identification of $\tilde{\delta}(x)$, follows by noting that:

\begin{align*}
&P(Y(0, M(0)) = 0 \mid Y = 1, M = 1, A = 1, x) \\
&= P(Y(0, M(0)) = 0 \mid Y(1, 1) = 1, M(1) = 1, A = 1, x) \\
&= P(Y(0, M(0)) = 0 \mid Y(1, 1) = 1, M(1) = 1, x) \\
&= P(Y(0, M(0)) = 0 \mid Y(1, M(1)) = 1, M(1) = 1, x)
\end{align*}

\noindent where the first equality holds by consistency, the second by Y-AM ignorability, and the final equality by consistency.
\end{remark}

\begin{remark}
In contrast to the positivity assumption in equation (\ref{ass:positivity}), we can allow that $P(Y = 1, M = 0 \mid x) = 0$ for any value of $x$.
\end{remark}

\begin{remark}
Identification of the total mediated probability of causation requires relatively weaker monotonicity assumptions than stated in equation (\ref{ass:monotonicity}), allowing for potentially negative direct and indirect effects, whereas assumption (\ref{ass:monotonicity}) essentially restricts that these effects are non-negative. We discuss the monotonicity conditions further in Remark \ref{rmk:monotonicity}.
\end{remark}

\subsubsection{Probability of indirect causation}

\begin{assumptionp}{\ref*{ass:aym}$^\star$}[A-YM ignorability]\label{ass:aym2}
\begin{align*}
A \perp \{Y(1, 1), Y(1, 0), M(0), M(1)\} \mid X, \qquad m = 0,1 \\
\end{align*}
\end{assumptionp}

\begin{assumptionp}{\ref*{ass:xworld}$^\star$}[Cross-world ignorability]
\begin{align*}
\{M(0), M(1)\} \perp \{Y(1, 0), Y(1, 1)\} \mid X, \qquad m = 0,1 \\
\end{align*}
\end{assumptionp}

\begin{assumptionp}{\ref*{ass:positivity}$^\star$}[A-M positivity]\label{ass:positivity2}
\begin{align*}
&P\{P(A = 1, M = m \mid X) \ge \epsilon\} = 1, \qquad \epsilon > 0, m = 0, 1 \\
&P\{P(A = 0 \mid X) \ge \epsilon\} = 1
\end{align*}
\end{assumptionp}

\begin{proof}[Proof of probability of indirect causation identification]

We invoke assumptions (\ref{ass:consistency}), (\ref{ass:monotonicity}), (\ref{ass:ypositivity}), and assumptions (\ref{ass:aym2})-(\ref{ass:positivity2}) above.

\begin{align}
\nonumber&P(Y(1, M(0)) = 0, Y(0, M(0)) = 0 \mid Y(1, M(1)) = 1, M(1) = 1, X = x) \\
\nonumber&= P(Y(1, M(0)) = 0 \mid Y(1, 1) = 1, M(1) = 1, X = x) \\
\nonumber&= \sum_m P(Y(1, M(0)) = 0 \mid Y(1, 1) = 1, M(1) = 1, M(0) = m, x) \\
\nonumber&\times P(M(0) = m \mid Y(1, 1) = 1, M(1) = 1, x) \\
\nonumber&= \sum_m P(Y(1, m) = 0 \mid Y(1, 1) = 1, M(1) = 1, M(0) = m, x)P(M(0) = m \mid M(1) = 1, x) \\
\label{eqn:ident1}&= P(Y(1, 0) = 0 \mid Y(1, 1) = 1, X = x)P(M(0) = 0 \mid M(1) = 1, X = x)
\end{align}

\noindent where the first equality follows by Y-M monotonicity and Y-M consistency, the second by the law of iterated expectations, the third by cross-world ignorability and Y-M consistency, and the fourth by the fact that $P(Y(1, 1) = 0 \mid Y(1, 1) = 1) = 0$ and cross-world ignorability. Now consider the first term:

\begin{align*}
&P(Y(1, 0) = 0 \mid Y(1, 1) = 1, X = x) \\
&=\frac{P(Y(1, 0) = 0, Y(1, 1) = 1 \mid X = x)}{P(Y(1, 1) = 1 \mid X = x)} \\
&=\frac{P(Y(1, 0) = 0, Y(1, 1) = 1 \mid X = x)}{P(Y(1, 1) = 1 \mid A = 1, M = 1, X = x)} \\
&=1 - \frac{P(Y(1, 0) = 1 \mid X = x)}{P(Y = 1 \mid A = 1, M = 1, X = x)} \\
&=1 - \frac{P(Y = 1 \mid A = 1, M = 0, X = x)}{P(Y = 1 \mid A = 1, M = 1, X = x)}
\end{align*}
where the first equality holds by Bayes' rule, the second by A-YM and cross-world ignorability, the third by Y-M monotonicity and Y-M consistency, and the final equality by A-YM and cross-world ignorability and Y-M consistency. The result follows from applying Proposition~\ref{prop:1} to the second term (which is simply the probability of causation with respect to the mediator) and multiplying the expressions.
\end{proof}

\begin{remark}
The equivalence of $\psi(x)$ and $\tilde{\psi}(x)$, and therefore the identification of $\tilde{\psi}(x)$ can be seen by noting that:
\begin{align*}
&P(Y(1, M(0)) = 0, Y(0, M(0)) = 0 \mid Y = 1, M = 1, A = 1, X = x) \\
&=P(Y(1, M(0)) = 0, Y(0, M(0)) = 0 \mid Y(1, 1) = 1, M(1) = 1, A = 1, X = x) \\
&=P(Y(1, M(0)) = 0 \mid Y(1, 1) = 1, M(1) = 1, X = x) \\
&=P(Y(1, M(0)) = 0, Y(0, M(0)) = 0 \mid Y(1, M(1)) = 1, M(1) = 1, X = x) \\
\end{align*}
\noindent where the first equality uses consistency, the second monotonicity and A-YM ignorability, and the final equality again uses consistency and monotonicity.
\end{remark}

\begin{remark}
The positivity requirement for the identification of $\psi(x)$ allows that for any $x$, either $P(A = 0, M = 1 \mid x) = 0$ or $P(A = 0, M = 0 \mid x) = 0$ (but not both, as we do require $P(A = 0 \mid x) \ge \epsilon$ for all $x$).
\end{remark}

\begin{remark}
Identification of $\zeta(x)$ follows directly from the proof of Theorem \ref{theorem1}, since $\delta(x) = \tau(x) + \zeta(x)$ by definition, and therefore the identification result requires the union of two sets of identifying assumptions outlined above.
\end{remark}

\begin{remark}\label{rmk:monotonicity}
One assumption that we do not require for identification, but may nevertheless be useful for interpretation, is a monotonicity condition that $Y(0, 1) \ge Y(0, 0)$. 

To see why, consider the case where for an individual $M(1) = 1$ and $M(0) = 0$; that is, the exposure induces the mediator; and we have the following potential outcomes: $Y(1, 1) = 1, Y(1, 0) = 1, Y(0, 1) = 0, Y(0, 0) = 1$. Under our definitions, we would say that for this individual, there was no effect in total, no indirect effect, and no direct effect. Yet in fact, for this individual, there is an indirect effect that acts only in the absence of the exposure (a so-called ``pure'' indirect effect), but a mediated interaction term that undoes this effect in the presence of the exposure \cite{vanderweele2014causal}. While mathematically this does not present a problem, this scenario may be conceptually challenging to allow for in applied examples. One can always further assume the monotonicity condition that $Y(0, 1) \ge Y(0, 0)$ to rule this out, if desired.
\end{remark}

\newpage

\subsection{Estimation}

\begin{proof}[Proof of Propositions \ref{proposition2}, \ref{proposition3}]
This proof follows directly from \cite{cuellar2020non} and applying the chain-rule to the functionals $\mathbb{E}[\psi(X)]$ and $\mathbb{E}[\delta(X)]$, treating $X$ as discrete \cite{kennedy2022eifs}. For example, for a fixed $\beta^\star$, we let $h(X; \beta^\star) = \frac{\partial g(X; \beta^\star)}{\partial \beta}w(X)$:

\begin{align*}
&\text{EIF}[\mathbb{E}[h(X; \beta^\star)(\psi(X) - g(X; \beta^\star)))]] \\
&= \text{EIF}[\sum_x h(x; \beta^\star)(\psi(x) - g(x; \beta^\star))p(x)] \\
&= \sum_x h(x; \beta^\star)\text{EIF}(\psi(x))p(x) + \sum_x w(x)(\psi(x) - g(x; \beta^\star))\text{EIF}(p(x)) \\
&= \sum_x h(x; \beta^\star)[\varphi_3(O; \eta) - \varphi_1(O; \eta) - \varphi_2(O; \eta)] + h(x; \beta^\star)(\psi(x) - g(x; \beta^\star))[1(X = x) - p(x)] \\
&= h(X; \beta^\star)[\varphi_3(O; \eta) - \varphi_1(O; \eta) - \varphi_2(O; \eta) + \psi(X) - g(X; \beta^\star)] - \Psi(\beta^\star) 
\end{align*}

\noindent where we again obtain $\text{EIF}[\psi(x)]$ by treating $x$ as discrete to obtain:

\begin{align*}
\text{EIF}[\psi(x)] &= \text{EIF}\left(1 - \frac{\mu_{10}(x)}{\mu_{11}(x)}\right)\left(1 - \frac{\gamma_0(x)}{\gamma_1(x)}\right) \\
&= \text{EIF}\left[\frac{\mu_{10}(x)}{\mu_{11}(x)}\frac{\gamma_0(x)}{\gamma_1(x)}\right] - \text{EIF}\left[\frac{\mu_{10}(x)}{\mu_{11}(x)}\right] - \text{EIF}\left[\frac{\gamma_0(x)}{\gamma_1(x)}\right] \\
&= \varphi_3(O; \eta) - \varphi_1(O; \eta) - \varphi_2(O; \eta)
\end{align*}

\noindent and applying the chain-rule to get each component part. To be precise, note that:

\begin{align*}
\text{EIF}\left[\frac{\mu_{10}(x)}{\mu_{11}(x)}\frac{\gamma_0(x)}{\gamma_1(x)}\right] &= \text{EIF}\left[\frac{\mu_{10}(x)}{\mu_{11}(x)}\right]\frac{\gamma_0(x)}{\gamma_1(x)} + \text{EIF}\left[\frac{\gamma_0(x)}{\gamma_1(x)}\right]\frac{\mu_{10}(x)}{\mu_{11}(x)} \\
&= \varphi_1(O; \eta)\frac{\gamma_0(x)}{\gamma_1(x)} + \varphi_2(O; \eta) \frac{\mu_{10}(x)}{\mu_{11}(x)}
\end{align*}

\noindent Second, by the quotient rule we obtain:

\begin{align*}
\varphi_1(O; \eta) &= \frac{1}{\mu_{11}(x)^2}(\text{EIF}[\mu_{10}(x)]\mu_{11}(x) - \text{EIF}[\mu_{11}(x)]\mu_{10}(x))  \\
\varphi_2(O; \eta) &= \frac{1}{\gamma_1(x)^2}(\text{EIF}[\gamma_0(x)]\gamma_1(x) - \text{EIF}[\gamma_1(x)]\gamma_0(x))
\end{align*}

\noindent Finally, recall that $\text{EIF}[E[Y \mid X = x]] = \frac{1(X = x)}{P(X = x)}[Y - E[Y \mid X = x]]$ \cite{kennedy2022eifs}. It follows that for any $(a, m, x)$:

\begin{align*}
\text{EIF}[\mu_{am}(x)] &= \frac{1(A = a, M = m, X = x)}{P(A = a, M = m, X = x)}(Y - \mu_{am}(X)), \\
\text{EIF}[\gamma_a(x)] &= \frac{1(A = a, X = x)}{P(A = a, X = x)}(M - \gamma_a(x))
\end{align*}

\noindent The result follows by combining expressions. The same logic can be used to derive the influence function for the moment conditions $\Psi_{ptc}(\alpha^\star)$ and $\Psi_{pnde}(\rho^\star)$, where for this final derivation we can use the fact that since $\zeta(x) = \delta(x) - \psi(x)$, then $\text{EIF}[\zeta(x)] = \text{EIF}[\delta(x)] - \text{EIF}[\psi(x)]$ \cite{kennedy2022eifs}. To flesh this out slightly more, we sketch the derivation of $\text{EIF}[\delta(x)]$ below:

\begin{align*}
\text{EIF}[\delta(x)] &= \text{EIF}\left(\left[1-\frac{\mu_{00}(x)}{\mu_{11}(x)}\right]\left[1 - \frac{\gamma_0(x)}{\gamma_1(x)}\right] + \left[1-\frac{\mu_{01}(x)}{\mu_{11}(x)}\right]\left[\frac{\gamma_0(x)}{\gamma_1(x)}\right]\right) \\
&= \text{EIF}\left(\frac{\mu_{00}(x)}{\mu_{11}(x)}\frac{\gamma_0(x)}{\gamma_1(x)} - \frac{\mu_{10}(x)}{\mu_{11}(x)}\frac{\gamma_0(x)}{\gamma_1(x)} - \frac{\mu_{00}(x)}{\mu_{11}(x)}\right) \\
&= \text{EIF}\left(\frac{\mu_{00}(x)}{\mu_{11}(x)}\right)\frac{\gamma_0(x)}{\gamma_1(x)} +
\frac{\mu_{00}(x)}{\mu_{11}(x)}\text{EIF}\left(\frac{\gamma_0(x)}{\gamma_1(x)}\right) \\
&- \text{EIF}\left(\frac{\mu_{10}(x)}{\mu_{11}(x)}\right)\frac{\gamma_0(x)}{\gamma_1(x)} - \text{EIF}\left(\frac{\gamma_0(x)}{\gamma_1(x)}\right)\frac{\mu_{10}(x)}{\mu_{11}(x)} - \text{EIF}\left(\frac{\mu_{00}(x)}{\mu_{11}(x)}\right) 
\end{align*}

\noindent By the quotient rule, we know generally that:

\begin{align*}
\text{EIF}\left(\frac{\mu_{am}(x)}{\mu_{a'm'}(x)}\right) &= \frac{1}{\mu_{a'm'}(x)^2}\left(\frac{I(A = a, M = m, X = x)}{P(A = a, M = m, X = x)}[Y - \mu_{am}(x)]\mu_{a'm'}(x)\right) \\
&- \left(\frac{I(A = a', M = m', X = x)}{P(A = a', M = m', X = x)}[Y - \mu_{a'm'}(x)]\mu_{am}(x)\right)
\end{align*}

\noindent Applying this to the relevant quantities above and plugging the previously derived result for $\text{EIF}\left[\frac{\gamma_0(x)}{\gamma_1(x)}\right]$ gives the result.

\end{proof}

\begin{proof}[Proof of Theorem~\ref{theorem3} (sketch)]
The proof follows almost directly Lemma 3 in \cite{kennedy2021semiparametric}. To formally complete the proof, we would also need to show that the quantity $P[\phi(O; \beta_0, \hat{\eta}) - \phi(O; \beta_0, \eta)]$ is second-order in the nuisance estimation to derive sufficient conditions where this quantity is $o_p(n^{-1/2})$, where $P[f(O)] = \mathbb{E}[f(O) \mid D_0^n]$ for a training dataset $D_0^n$ used to estimate the nuisance parameters $\eta$. The conditions for this to hold will then be satisfied when $\|\hat{\eta} - \eta\| = o_p(n^{-1/4})$.
\end{proof}

\newpage

\section{Probability of causation: other decompositions}\label{app:othres}

As noted in Section \ref{sec:identification}, we can also decompose the probability of causation decomposes into analogous terms that we described in our paper that do not condition on $M(1)$. We call these the total probabilities of indirect and direct causation and denote them $\alpha(x)$ and $\beta(x)$, respectively. More formally, we can derive these quantities by applying the law of iterated expectation the the probability of causation:

\begin{align*}
P(Y(0) = 0 \mid Y(1) = 1, x) &= P(Y(1,M(0)) = 0, Y(0, M(0)) = 0 \mid Y(1, M(1)) = 1, x) \\
&+ P(Y(1, M(0)) = 1, Y(0, M(0)) = 0) \mid Y(1, M(1)) = 1, x) \\
&= \alpha(x) + \beta(x)
\end{align*}

\noindent Proposition \ref{prop:10} shows that these expressions are identifiable in the observed data distribution given $O = (X, A, M, Y)$. 

\begin{proposition}[Identification of total indirect and direct probabilities of causation]\label{prop:10}
Under assumptions (\ref{ass:consistency})-(\ref{ass:ypositivity}), the total probabilities of indirect and direct causation can be expressed as:
\begin{align*}
\alpha(x) &= \left(1 - \frac{\mu_{10}(x)}{\mu_{11}(x)}\right)\left(1 - \frac{\gamma_0(x)}{\gamma_1(x)}\right)P(M = 1 \mid A = 1, Y = 1, x) \\
\beta(x) &= \left(1 - \frac{\mu_0(x)}{\mu_1(x)}\right) - \alpha(x)
\end{align*}
\noindent where $\mu_a(x) = P(Y = 1 \mid A = a, x)$.
\end{proposition}

\begin{remark}\label{rmk:1}
The identifying expression for $\alpha(x)$ is simply the identifying expression for $\psi(x)$ scaled by the identifying expression for $P(M(1) = 1 \mid Y(1) = 1, x)$. This result is highly intuitive: under monotonicity, an indirect effect can only occur on the stratum where $M(1) = 1$: if $M(1) = 0$, an effect can only have been induced via the direct $A \to Y$ pathway. Therefore, the probability of indirect causation is simply a scaled version of the probability of indirect causation. The total probability of direct causation is then simply the probability of causation minus the total probability of indirect causation.
\end{remark}

\begin{remark}
When we do not observe $M$, we cannot identify these expressions given the observed data $(X, A, Y)$. On the other hand, once we observe $M$, we may able to condition on this event, and so, as argued in the text, the expressions $\alpha(x)$ and $\beta(x)$ have no obvious practical relevance.
\end{remark}

\begin{remark}
We can also define terms $\tilde{\alpha}(x)$ and $\tilde{\beta}(x)$ that condition on the exposure and the observed rather than potential outcomes. These quantities have the same identifying expression under our assumptions.
\end{remark}

We also derived terms in Section \ref{sec:identification} that were analogous to the total mediated probability of causation $\delta(x)$, but defined on the stratum where $M(1) = 0$ rather than the stratum where $M(1) = 1$. We called this term $\delta'(x)$. We formally consider identification of this term; however, we first note that by the law of iterated expectations, we can similarly decompose this term into an analogous indirect and direct probability of causation, $\psi'(x)$ and $\zeta'(x)$.

\begin{align*}
\delta'(x) &= P(Y(0) = 0 \mid Y(1, M(1)) = 1, M(1) = 0, x) \\
&= P(Y(1, M(0)) = 0, Y(0, M(0)) = 0 \mid Y(1, M(1)) = 1, M(1) = 0, X = x) \\
&+ P(Y(1, M(0)) = 1, Y(0, M(0)) = 0 \mid Y(1, M(1)) = 1, M(1) = 0, X = x) \\
&= \psi'(x) + \zeta'(x)
\end{align*}

\noindent Proposition \ref{prop:11} presents an identification result for these three terms. The proofs of both propositions \ref{prop:10} and \ref{prop:11} are below.

\begin{proposition}[Identification of $\delta'(x), \psi'(x), \zeta'(x)$]\label{prop:11}
    Under assumptions (\ref{ass:consistency})-(\ref{ass:ypositivity}), $\delta'(x)$, $\psi'(x)$, and $\zeta'(x)$ are identified by the following expressions:
    \begin{align*}
        \zeta'(x) &= \left(1 - \frac{\mu_{00}(x)}{\mu_{10}(x)}\right) \\
        \delta'(x) &= \zeta'(x) \\
        \psi'(x) &= 0
    \end{align*}
\end{proposition}   

\begin{proof}[Proof of Proposition \ref{prop:10}]
\begin{align*}
\alpha(x) &= \sum_{m_0, m_1 = 0, 1}P(Y(1, m_0) = 0, Y(0, m_0) = 0 \mid M(1) = m_1, M(0) = m_0, Y(1, m_1) = 1, x) \\
&\times P(M(0) = m_0, M(1) = m_1 \mid Y(1) = 1, x) \\
&= \sum_{m_0, m_1 = 0, 1}P(Y(1, m_0) = 0 \mid M(1) = m_1, M(0) = m_0, Y(1, m_1) = 1, x) \\
&\times P(M(0) = m_0, M(1) = m_1 \mid Y(1) = 1, x) \\
&= \sum_{m_0, m_1 = 0, 1}P(Y(1, m_0) = 0 \mid Y(1, m_1) = 1, x)P(M(0) = m_0, M(1) = m_1 \mid Y(1) = 1, x) \\
&= P(Y(1, 0) = 0 \mid Y(1, 1) = 1, x)P(M(0) = 0, M(1) = 1 \mid Y(1) = 1, x) \\
&= P(Y(1, 0) = 0 \mid Y(1, 1) = 1, x)P(M(0) = 0 \mid M(1) = 1, Y(1, 1) = 1, x)P(M(1) = 1 \mid Y(1) = 1, x) \\
&= \underbrace{P(Y(1, 0) = 0 \mid Y(1, 1) = 1, x)P(M(0) = 0 \mid M(1) = 1, x)}_{T_1}\underbrace{P(M(1) = 1 \mid Y(1) = 1, x)}_{T_2} 
\end{align*}

\noindent where the first equality follows by iterating expectations over the conditional distributions of $M(1)$ and $M(0)$ and consistency, the second by monotonicity, and the third by cross-world ignorability. The fourth equality follows from the fact that when $m_1 = m_0$, $P(Y(1, m_0) = 0 \mid Y(1, m_1) = 1) = 0$, and by monotonicity, $P(M(1) = 0, M(0) = 1) = 0$. The fifth equality follows by definition, and the sixth by cross-world ignorability. The term $T_1$ has previously been identified in the proof of identification of $\psi(x)$; the final results holds by by applying A-YM ignorability and consistency to $T_2$ (noting that we have used positivity assumptions throughout). Identification of $\beta(x)$ follows from the fact that $\alpha(x) + \beta(x) = \tau(x)$ and applying Proposition $\ref{prop:1}$ to identify $\tau(x)$.
\end{proof}

\begin{proof}[Proof of Proposition \ref{prop:11}]
\begin{align*}
\psi'(x) &=P(Y(1, M(0)) = 0, Y(0, M(0)) = 0 \mid Y(1, M(1)) = 1, M(1) = 0, X = x) \\
&=P(Y(1, M(0)) = 0, Y(0, M(0)) = 0 \mid Y(1, 0) = 1, M(1) = 0, X = x) \\
&=P(Y(1, 0) = 0, Y(0, 0) = 0 \mid Y(1, 0) = 1, M(1) = 0, X = x) \\
&=0
\end{align*}

\noindent where the first equality follows by definition, the second by consistency, the third by monotonicity ($M(1) = 0 \implies M(0) = 0$) and consistency, and the final equality by the fact that $P(Y(1, 0) = 0 \mid Y(1, 0) = 1) = 0$. Next, we consider $\zeta'(x)$:

\begin{align*}
\zeta'(x) &=P(Y(1, M(0)) = 1, Y(0, M(0)) = 0 \mid Y(1, M(1)) = 1, M(1) = 0, X = x) \\
&=P(Y(1, 0) = 1, Y(0, 0) = 0 \mid Y(1, 0) = 1, M(1) = 0, X = x) \\
&=P(Y(0, 0) = 0 \mid Y(1, 0) = 1, M(1) = 0, X = x) \\
&=P(Y(0, 0) = 0 \mid Y(1, 0) = 1, X = x) \\
&=\left(1 - \frac{\mu_{00}(x)}{\mu_{10}(x)}\right)
\end{align*}

\noindent where the first equality follows by definition, the second by consistency and monotonicity, the third by the fact that $P(Y(1, 0) = 1 \mid Y(1, 0) = 1, M(1) = 0, x) = 1$, the fourth by cross-world ignorability, and the final follows from the proof of Theorem \ref{theorem1}. Since $\delta'(x) = \zeta'(x) + \psi'(x)$, the final result follows.
\end{proof}

\begin{remark}
As noted both in Section \ref{sec:identification} and Remark \ref{rmk:1}, the probability of indirect causation on this stratum is equal to zero, essentially due to the monotonicity requirements. As a result, the total mediated probability of causation on this stratum is equal to the probability of direct causation.
\end{remark}

\newpage

\section{Simulation specifications}\label{app:specs}

We describe the data generating processes for the simulation studies. First, for each simulation we make $n = 1,000$ i.i.d. draws of a single covariate $X \sim \mathcal{U}(0, 1)$. Next, we generate the propensity scores, potential mediators, and potential outcomes from the following functions: 

\begin{align*}
    \pi(X) = \text{expit}(-1 - 0.5X) \\
    \gamma_0(X) = 0.25 - 0.1X - 0.1X^2 \\
    \tilde{\gamma}_1(X) = 0.65 - 0.1X - 0.2X^2 \\
    \mu_{00}(X) = 0.5 - 0.1X - 0.2X^2 \\
    \mu_{01}(X) = 0.3 - 0.1X - 0.1X^2 \\
    \tilde{\mu}_{10}(X) = 0.6 - 0.1X - 0.1X^2 \\
    \tilde{\mu}_{11}(X) = 0.4 - 0.1X - 0.2X^2
\end{align*}

\noindent We then draw the exposure, potential mediators, and potential outcomes as follows:

\begin{align*}
    A \mid X &\sim \text{Bern}(\pi(X)) \\
    M(0) \mid X &\sim \text{Bern}(\gamma_0(X)) \\
    M(1) \mid X, M(0) = 0 &\sim \text{Bern}(\tilde{\gamma}_1(X)) \\
    Y(0, 0) \mid X &\sim \text{Bern}(\mu_{00}(X)) \\
    Y(0, 1) \mid X &\sim \text{Bern}(\mu_{01}(X)) \\
    Y(1, 0) \mid X, Y(0, 0) = 0 &\sim \text{Bern}(\tilde{\mu}_{10}(X)) \\
    Y(1, 1) \mid X, Y(0, 1) = 0, Y(1, 0) = 0 &\sim \text{Bern}(\tilde{\mu}_{11}(X))
\end{align*}

\noindent where we enforce the monotonicity constraints so that, for example, $M(0) = 1 \implies M(1) = 1$, $Y(0, 0) = 1 \implies Y(1, 0) = 1, Y(0, 1) = 1, Y(1, 1) = 1$, and so forth. This implies that:

\begin{align*}
    \gamma_1(X) &= \tilde{\gamma}_1(X)(1 - \gamma_0(X)) + \gamma_0(X) \\
    \mu_{10}(X) &= \tilde{\mu}_{10}(X)(1 - \mu_{00}(X)) + \mu_{00}(X) \\
    \mu_{11}(X) &= \tilde{\mu}_{11}(X)((1 - \mu_{01}(X))(1 - \tilde{\mu}_{10}(X))(1 - \mu_{00}(X))) \\ 
    &+ 1 - ((1 - \mu_{01}(X))(1 - \tilde{\mu}_{10}(X))(1 - \mu_{00}(X)))
\end{align*}

\noindent Finally, applying consistency we set:

\begin{align*}
M &= AM(1) + (1-A)M(0) \\
Y &= AY(1) + (1-A)Y(0) \\
Y(1) &= Y(1, 1)M(1) + Y(1, 0)(1-M(1)) \\
Y(0) &= Y(0, 1)M(0) + Y(0, 0)(1-M(0)) \\
Y(1, M(0)) &= Y(1, 1)M(0) + Y(1, 0)(1-M(0)) \\
Y(0, M(1)) &= Y(0, 1)M(1) + Y(0, 0)(1-M(1))
\end{align*}

\end{document}